\documentclass[onecolumn]{IEEEtran}
\usepackage{amsmath}
\usepackage{amssymb}
\usepackage{mathrsfs}
\usepackage{cite}
\usepackage{epsfig}
\usepackage{graphics}
\usepackage{hyperref}

\newtheorem{thm}{Theorem}
\newtheorem{lem}{Lemma}
\newtheorem{coro}{Corollary}
\newtheorem{defi}{Definition}

\newtheorem{rem}{Remark}

\begin{document}

\title{On the Secure Degrees-of-Freedom of the Multiple-Access-Channel}

\author{Ghadamali Bagherikaram, Abolfazl S. Motahari, Amir K. Khandani\\
Coding and Signal Transmission Laboratory,
 Department of
Electrical
and Computer Engineering,\\
 University of Waterloo, Ontario,
 N2L 3G1,
 Emails: \{gbagheri,abolfazl,khandani\}@cst.uwaterloo.ca
}
 \maketitle
\footnotetext{Financial support provided by Nortel, and the
corresponding matching funds by the Natural Sciences and Engineering
Research Council of Canada (NSERC), and Ontario Ministry of Research
\& Innovation (ORF-RE) are gratefully acknowledged.}
\begin{abstract}
A $K$-user secure Gaussian Multiple-Access-Channel (MAC) with an
external eavesdropper is considered in this paper. An achievable
rate region is established for the secure discrete memoryless MAC.
The secrecy sum capacity of the degraded Gaussian MIMO MAC is proven
using Gaussian codebooks. For the non-degraded Gaussian MIMO MAC, an
algorithm inspired by interference alignment technique is proposed
to achieve the largest possible total Secure-Degrees-of-Freedom
(S-DoF). When all the terminals are equipped with a single antenna,
Gaussian codebooks have shown to be inefficient in providing a
positive S-DoF. Instead, a novel secure coding scheme is proposed to
achieve a positive S-DoF in the single antenna MAC. This scheme
converts the single-antenna system into a multiple-dimension system
with fractional dimensions. The achievability scheme is based on the
alignment of signals into a small sub-space at the eavesdropper, and
the simultaneous separation of the signals at the intended receiver.
Tools from the field of Diophantine Approximation in number theory
are used to analyze the probability of error in the coding scheme.
It is proven that the total S-DoF of $\frac{K-1}{K}$ can be achieved
for almost all channel gains. For the other channel gains, a
multi-layer coding scheme is proposed to achieve a positive S-DoF.
As a function of channel gains, therefore, the achievable S-DoF is
discontinued.
\end{abstract}
\section{Introduction}
The notion of information theoretic secrecy in communication systems
was first introduced by Shannon in \cite{1}. The information
theoretic secrecy requires that the received signal of the
eavesdropper not provide any information about the transmitted
messages. Following the pioneering works of Wyner \cite{2} and
Csiszar et. al. \cite{3} which studied the wiretap channel, many
multi-user channel models have been considered from a perfect
secrecy point of view. The secrecy capacity rate has been
established for some basic models, including the Gaussian wiretap
channel \cite{4}, the MIMO wiretap channel \cite{5,6,7,8,9}, and the
fading wiretap channel \cite{10,11}. Additionally, the secrecy
capacity rate regions have been characterized successfully for some
multi-user channels, including the MIMO Gaussian broadcast channel
with confidential messages \cite{12,13,14,15,16}.

The secure Gaussian MAC with/without an external eavesdropper is
introduced in \cite{17,18,19,20}. The secure Gaussian MAC with an
external eavesdropper consists of an ordinary Gaussian MAC and an
external eavesdropper. The capacity region of this channel is still
an open problem in the information theory field. For this channel,
an achievable rate scheme based on Gaussian codebooks is proposed in
\cite{20}, and  the sum secrecy capacity of the degraded Gaussian
channel is found in \cite{18}. For some special cases, upper bounds,
lower bounds, and some asymptotic results on the secrecy capacity
exist; see for example \cite{21,22,23,24}. For the achievability
part, Shannon's random coding argument has proven to be effective in
these works.

The secure MAC generalizes the wiretap channel. In the wiretap
channel, the direct coding scheme uses the framework of random
coding, which is widely used in the analysis of multi-terminal
source and channel coding problems. One approach to find achievable
sum rates for the secure MAC is to extend the random coding solution
to the secure multi-user. As we will show in this paper, this
extension leads to a single-letter characterization for the secure
rate region of the MAC. Our achievability, as usual, is based on the
i.i.d random binning scheme.

On the other hand, it is shown that the random coding argument may
be insufficient to prove capacity theorems for certain channels;
instead, structure codes can be used to construct efficient channel
codes for Gaussian channels. In reference \cite{25}, nested lattice
codes are used to provide secrecy in two-user Gaussian channels. In
\cite{25} it is shown that structure codes can achieve a positive
S-DoF in a two-user MAC. In particular, the achievability scheme of
\cite{25} provides an S-DoF of $\frac{1}{2}$ for a small category of
channel gains and for the other categories, it provides a S-DoF of
strictly less than $\frac{1}{2}$.

In reference \cite{26}, the concept of interference alignment is
introduced and has illustrated its capability in achieving the full
DoF of a class of two-user X channels. In reference \cite{27}, and
\cite{28}, a novel coding scheme applicable in networks with single
antenna nodes is proposed. This scheme converts a single antenna
system into an equivalent Multiple Input Multiple Output (MIMO)
system with fractional dimensions.

In this work \footnotetext{Please see the conference version of this
work in \cite{29}.} we establish the secrecy sum capacity of the
degraded Gaussian MIMO MAC using random bininng of Gaussian
codebooks. For the non-degraded channel, we present an algorithm
inspiring the notion of signal alignment to achieve the largest
S-DoF by using Gaussian codebooks. We then use the notion of
\emph{real alignment} of \cite{27} to prove that for almost all
channel gains in the secure $K$ user single-antenna Gaussian MAC, we
can achieve the S-DoF of $\frac{K-1}{K}$. Here, our scheme uses
structure codes instead of Gaussian codebooks. In the case of the
channel gains for which the S-DoF of $\frac{K-1}{K}$ cannot be
achieved, we propose a multi-layer coding scheme to achieve a
positive S-DoF. The scheme of this work differs from that of
\cite{25}, in the sense that our scheme achieves the S-DoF of
$\frac{1}{2}$ for almost all channel gains, while reference
\cite{25} provides a scheme that achieves the S-DoF of $\frac{1}{2}$
(when $K=2$) for  specific channel gains, i.e., algebraic irrational
gains. Therefore, we prove that the carve of S-DoF versus channel
gains is almost certainly constant with many discontinuations.

The rest of the paper is organized as follows: Section II provides
some background and preliminaries. In section III, we consider the
secure Gaussian MIMO MAC and establish the secrecy sum rate for the
degraded channel. In this section, we propose an algorithm to
achieve the largest possible value of S-DoF for the non-degraded
channel model. We present our results for the achievable S-DoF of
the single-antenna MAC in section IV. Finally, section V concludes
the paper.

\section{Preliminaries}
Consider a secure $K$-user Gaussian MIMO Multiple-Access-Channel
(MAC) as depicted in Fig. \ref{f1}.
\begin{figure}
\centerline{\includegraphics[scale=1.1]{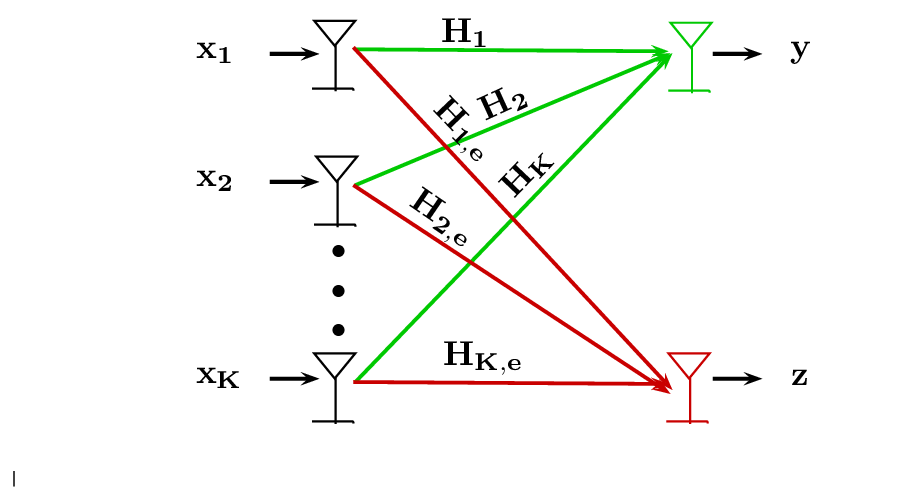}} \caption{Secure $K$-user Gaussian MIMO Multiple-Access-Channel} \label{f1}
\end{figure}
In this confidential setting, each user $k$ ($k\in
\mathcal{K}\stackrel{\triangle}{=}\{1,2,...,K\}$) wishes to send a
message $W_{k}$ to the intended receiver in $n$ uses of the channel
simultaneously, and prevent the eavesdropper from having any
information about the messages. At a specific time, the signals
received by the intended receiver and the eavesdropper is given by
\begin{IEEEeqnarray}{lr}\label{eq0}
\mathbf{y}&=\sum_{k=1}^{K}\mathbf{H_{k}}\mathbf{x_{k}}+\mathbf{n_{1}}\\ \nonumber
\mathbf{z}&=\sum_{k=1}^{K}\mathbf{H_{k,e}}\mathbf{x_{k}}+\mathbf{n_{2}}
\end{IEEEeqnarray}
where
\begin{itemize}
  \item $\mathbf{x_{k}}$ for $k=1,2,...,K$ is a real input vector of size $M_{k}\times 1$ under an input average power constraint. We require that $Tr(\mathbf{Q_{k}})\leq P$, where $\mathbf{Q_{k}}=E[\mathbf{x_{k}}\mathbf{x_{k}^{\dag}}]$. Here, the superscript
$\dag$ denotes the Hermitian transpose of a vector and $Tr(.)$ denotes the Trace operator on the  matrices.
  \item $\mathbf{y}$ and $\mathbf{z}$ are real output vectors  which are received by the destination and the eavesdropper, respectively. These are vectors of size $N \times 1$ and $N_{e} \times 1$, respectively.
  \item $\mathbf{H_{k}}$ and $\mathbf{H_{k,e}}$ for $k=1,2,...,K$ are fixed, real gain matrices which model the channel gains between the transmitters and the intended receiver, and the eavesdropper, respectively. These are matrices of size $N \times M_{k}$ and $N_{e} \times M_{k}$, respectively. The channel state information is assumed to be known perfectly at all the transmitters and at all receivers.
  \item $\mathbf{n_{1}}$ and $\mathbf{n_{2}}$  are real Gaussian random vectors with zero means and covariance matrices $\mathbf{N_{1}}=E[\mathbf{n_{1}}\mathbf{n_{1}}^{T}]=\mathbf{I_{N}}$ and $\mathbf{N_{2}}=E[\mathbf{n_{2}}\mathbf{n_{2}}^{T}]=\mathbf{I_{N_{e}}}$, respectively. Here, $\mathbf{I_{M}}$ represents the identity matrix of size $M\times M$.
\end{itemize}
Let $\mathbf{x_{k}^{n}}$, $\mathbf{y^{n}}$ and $\mathbf{z^{n}}$
denote the random channel inputs and random channel outputs matrices
over a block of $n$ samples. Furthermore, let $\mathbf{n_{1}^{n}}$,
and $\mathbf{n_{2}^{n}}$ denote the additive noises of the channels.
Therefore, we have
\begin{IEEEeqnarray}{rl}
\mathbf{y^{n}}&=\sum_{k=1}^{K}\mathbf{H_{k}}\mathbf{x_{k}^{n}}+\mathbf{n_{1}^{n}}\\
\nonumber
\mathbf{z^{n}}&=\sum_{k=1}^{K}\mathbf{H_{k,e}}\mathbf{x_{k}^{n}}+\mathbf{n_{2}^{n}}.
\end{IEEEeqnarray}
Note that bold vectors are random while the matrices $\mathbf{H_{k}}$ and $\mathbf{H_{k,e}}$ are deterministic matrices for all $k\in\mathcal{K}$. The columns of
$\mathbf{n_{1}^{n}}$ and $\mathbf{n_{2}^{n}}$ are independent Gaussian random vectors with covariance matrices $\mathbf{I_{N}}$ and $\mathbf{I_{N_{e}}}$, respectively. In
addition $\mathbf{n_{1}^{n}}$ and $\mathbf{n_{2}^{n}}$ are independent of $\mathbf{x_{k}^{n}}$'s and $W_{k}$'s. A
$((2^{nR_{1}},2^{nR_{2}},...,2^{nR_{k}}),n)$ secret code for the above channel consists of
the following components:

\emph{1}) $K$ secret message sets $\mathcal{W}_{k}=\{1,2,...,2^{nR_{k}}\}$.

\emph{2}) $K$  stochastic encoding functions $f_{k}(.)$ which map the secret messages to the transmitted symbols, i.e., $f_{k}: w_{k}\rightarrow \mathbf{x_{k}^{n}}$ for each $w_{k}\in\mathcal{W}_{k}$. At encoder $k$, each codeword is designed according to the transmitter's average power constraint $P$.

\emph{3}) A decoding function $\phi(.)$ which maps the received
symbols to estimate the messages: $\phi(\mathbf{y^{n}}) \rightarrow
(\hat{W_{1}},...,\hat{W_{K}}) $.

The reliability of the transmission is measured by the average
probability of error, which is defined as the probability that the
decoded messages are not equal to the transmitted messages; that is,
\begin{equation}
P_{e}^{(n)}=\frac{1}{\prod_{k=1}^{K}2^{nR_{k}}}\sum_{(w_{1},...,w_{K})\in\mathcal{W}_{1}\times....\times\mathcal{W_{K}}}P(\phi(\mathbf{y^{n}})\neq (w_{1},...,w_{K})|(w_{1},...,w_{k})~ \hbox{is sent}).
\end{equation}
The secrecy level is measured by the normalized equivocation defined
as follows: The normalized equivocation for each subset of messages
$W_{\mathcal{S}}$ for $\mathcal{S}\subseteq\mathcal{K}$ is
\begin{equation}
\Delta_{\mathcal{S}}\stackrel{\triangle}{=}\frac{H(W_{\mathcal{S}}|\mathbf{z^{n}})}{H(W_{\mathcal{S}})}.
\end{equation}
The rate-equivocation tuple $(R_{1},...,R_{K},d)$ is said to be
achievable for the Gaussian MIMO Multiple-Access-Channel with
confidential messages, if for any $\epsilon>0$, there exists a
sequence of $((2^{nR_{1}},...,2^{nR_{K}}),n)$ secret codes, such
that for sufficiently large $n$,
\begin{equation}
P_{e}^{(n)}\leq\epsilon,
\end{equation}
and
\begin{equation}
\Delta_{\mathcal{S}}\geq d-\epsilon,~~~~~~~\forall\mathcal{S}\subseteq\mathcal{K}.
\end{equation}
The perfect secrecy rate tuple $(R_{1},...,R_{K})$ is said to be achievable when $d=1$. When all the transmitted messages are perfectly secure, we have
\begin{equation}\label{eq1}
\Delta_{\mathcal{K}}\geq 1-\epsilon,
\end{equation}
or equivalently
\begin{equation}
H(W_{\mathcal{K}}|\mathbf{z^{n}})\geq H(W_{\mathcal{K}})-\epsilon H(W_{\mathcal{K}}).
\end{equation}
The normalized equivocation of each subset of messages can then be
written as follows:
\begin{IEEEeqnarray}{rl}
H(W_{\mathcal{S}}|\mathbf{z^{n}})&\stackrel{(a)}{=}H(W_{\mathcal{S}},W_{\mathcal{S}^{c}}|\mathbf{z^{n}})-H(W_{\mathcal{S}^{c}}|W_{\mathcal{S}},\mathbf{z^{n}})\\ \nonumber &=
H(W_{\mathcal{K}}|\mathbf{z^{n}})-H(W_{\mathcal{S}^{c}}|W_{\mathcal{S}},\mathbf{z^{n}})\\ \nonumber &\stackrel{(b)}{\geq}H(W_{\mathcal{K}})-\epsilon H(W_{\mathcal{K}})-H(W_{\mathcal{S}^{c}}|W_{\mathcal{S}},\mathbf{z^{n}})\\ \nonumber &\stackrel{(c)}{=}H(W_{\mathcal{S}})+H(W_{\mathcal{S}^{c}}|W_{\mathcal{S}})-\epsilon H(W_{\mathcal{K}})-H(W_{\mathcal{S}^{c}}|W_{\mathcal{S}},\mathbf{z^{n}})\\ \nonumber &\stackrel{(d)}{\geq}H(W_{\mathcal{S}})-\epsilon H(W_{\mathcal{K}}),
\end{IEEEeqnarray}
where $(a)$ and $(c)$ follow from the chain rule, $(b)$ follows from (\ref{eq1}) and $(d)$ follows from the fact that conditioning always decreases the amount of entropy.
Therefore, the normalized equivocation of each subset of messages is
\begin{equation}
\Delta_{\mathcal{S}}\geq 1-\epsilon^{'},
\end{equation}
where
$\epsilon^{'}=\frac{H(W_{\mathcal{K}})}{H(W_{\mathcal{S}})}\epsilon$.
Thus, when all of the $K$ messages are perfectly secure then it is
guaranteed that any subset of the messages becomes perfectly secure.

The total Secure Degrees-of Freedom (S-DoF) of $\eta$ is said to be
achievable, if the rate-equivocation tuple $(R_{1},...,R_{K},d=1)$
is achievable, and
\begin{equation}\label{sdf}
\eta=\lim_{P\rightarrow\infty}\frac{\sum_{k=1}^{K}R_{k}}{\frac{1}{2}\log P}
\end{equation}
\section{Secure DoF of the Multiple-Antenna Multiple-Access-Channel}
In this section, we first present an achievability rate region for
the secure discrete memoryless MAC. We then characterize the sum
capacity of the degraded secure discrete memoryless and degraded
Gaussian MIMO MAC. We present an achievable S-DoF of the
non-degraded Gaussian MIMO MAC under the perfect secrecy constraint
using Gaussian codebooks. In order to satisfy the perfect secrecy
constraint, we use the random binning coding scheme to generate the
codebooks. To maximize the achievable degrees of freedom, we adopt
the signal alignment scheme to separate the signals at the intended
receiver and simultaneously align the signals into a small subspace
at the eavesdropper.

\subsection{Discrete Memoryless MAC}
In this subsection, we study the secure discrete MAC of
$P(y,z|x_{1},...,x_{K})$ with $K$ users and an external
eavesdropper. The following theorems illustrate our results:
\begin{thm}\label{th1}
For the perfectly secure discrete memoryless MAC of $P(y,z|x_{1},...,x_{K})$, the region of
\begin{IEEEeqnarray}{rl}\label{eq2}
\left\{(R_{1},...,R_{K})|\sum_{i\in\mathcal{S}}R_{i}\leq
I(U_{\mathcal{S}};Y|U_{\mathcal{S}^{c}}),\sum_{k\in\mathcal{K}}R_{k}\leq\left[I(U_{\mathcal{K}};Y)-I(U_{\mathcal{K}};Z)\right]^{+},~~~\forall
\mathcal{S}\subset\mathcal{K}\right\},
\end{IEEEeqnarray}
for any distribution of
$P(u_{1})P(u_{2})....P(u_{K})P(x_{1}|u_{1})P(x_{2}|u_{2})...P(x_{K}|u_{K})P(y,z|x_{1},...,x_{K})$,
is achievable.
\end{thm}
\begin{proof}
The proof is available in Appendix $A$.
\end{proof}
In this Theorem $[x]^{+}$ denotes the positivity operator, i.e.,
$[x]^{+}=\max(x,0)$. Reference \cite{20} derived an achievable rate
region with Gaussian codebooks and power control for the Gaussian
secure MAC when all the transmitters and receivers are equipped with
a single antenna. Theorem \ref{th1}, however, gives an achievability
secrecy rate region for the general discrete memoryless MAC. Our
achievability rate region is also larger than the region of
\cite{20} in the special Gaussian channel case. Therefore, we have
the following achievable sum rate for the secure discrete memoryless
MAC:

\begin{coro}\label{cor2}
For the secure discrete memoryless MAC of $P(y,z|x_{1},...,x_{K})$,
the following sum rate is achievable:
\begin{IEEEeqnarray}{rl}
R_{\hbox{sum}}=\max\left[I(U_{\mathcal{K}};Y)-I(U_{\mathcal{K}};Z)\right]^{+},
\end{IEEEeqnarray}
where the maximization is over all distributions
$P(u_{1})...P(u_{K})P(x_{1}|u_{1})...P(x_{K}|u_{K})P(y,z|x_{1},...,x_{K})$
that satisfy the markov chain $W_{\mathcal{K}}\rightarrow
U_{\mathcal{K}}\rightarrow X_{\mathcal{K}}\rightarrow YZ$.
\end{coro}

\subsection{Gaussian MIMO MAC}
Consider the secure Gaussian MIMO MAC of (\ref{eq0}) which can be re-written as follows:
\begin{IEEEeqnarray}{lr}
\mathbf{y}&=\mathbf{H}\mathbf{x}+\mathbf{n_{1}}\\ \nonumber
\mathbf{z}&=\mathbf{H_{e}}\mathbf{x}+\mathbf{n_{2}}
\end{IEEEeqnarray}
where, $\mathbf{H}=\left[\mathbf{H_{1}}, \mathbf{H_{2}},...,
\mathbf{H_{K}}\right]$, $\mathbf{H_{e}}=\left[\mathbf{H_{1,e}},
\mathbf{H_{2,e}},..., \mathbf{H_{K,e}}\right]$, and
$\mathbf{x}=\left[\mathbf{x_{1}^{\dag}}, \mathbf{x_{2}^{\dag}},...,
\mathbf{x_{K}^{\dag}}\right]^{\dag}$. Without loss of generality,
assume that all nodes are equipped with the same number of antennas,
i.e., $M_{k}=N=N_{e}$ for all $k\in\mathcal{K}$. Note that when the
channel gain matrices $\mathbf{H_{k}}$ and $\mathbf{H_{k,e}}$ are
identity matrices we can determine that one channel output is
degraded w.r.t. another by examining whether their noise covariances
can be ordered correctly. In (\ref{eq0}), however, all noise
covariances are identity matrices and the receive vectors differ
only in their channel gain matrices. Therefore, similar to
\cite{30}, we use the following definition to determine a
degradedness order:
\begin{defi}
A receive vector
$\mathbf{z}=\mathbf{H_{e}}\mathbf{x}+\mathbf{n_{2}}$ is said to be
degraded w.r.t. $\mathbf{y}=\mathbf{H}\mathbf{x}+\mathbf{n_{1}}$ if
there exists a matrix $\mathbf{D}$ such that
$\mathbf{DH}=\mathbf{H_{e}}$ and such that
$\mathbf{D}\mathbf{D^{\dag}}\preceq\mathbf{I}$. Alternatively, we
say that $\mathbf{H_{e}}$ is degraded w.r.t. $\mathbf{H}$.
\end{defi}

According to this definition, it is easy to see that $\mathbf{y}$
can be approximated by multiplying $\mathbf{D}\mathbf{y}$. The
approximated channel has a different additive noise which is now
given by $\mathbf{Dn_{1}}\sim \mathcal{N}(0,\mathbf{DD^{\dag}} )$
compared to the original channel. As this approximated channel has
less noise ($\mathbf{DD^{\dag}} \preceq \mathbf{I}$), however, it is
clear that any message that can be decoded by the eavesdropper, can
also be decoded by the intended receiver. In the other words
$W_{\mathcal{K}}\rightarrow \mathbf{x}\rightarrow \mathbf{y}
\rightarrow \mathbf{z}$ forms a Markov chain.

\begin{thm}\label{th2}
The secrecy sum capacity of the degraded Gaussian MIMO MAC is given by
\begin{IEEEeqnarray}{lr}
C_{\hbox{sum}}=\max_{\mathbf{Q_{k}}: \mathbf{Q_{k}}\succeq 0, Tr(\mathbf{Q_{k}})\leq P}\frac{1}{2}\log\left|I+\sum_{k\in\mathcal{K}}\mathbf{H_{k}Q_{k}H_{k}^{\dag}}\right|-\frac{1}{2}\log\left|I+\sum_{k\in\mathcal{K}}\mathbf{H_{k,e}Q_{k}H_{k,e}^{\dag}}\right|.
\end{IEEEeqnarray}
\end{thm}
\begin{proof}
We need to show that the secrecy sum capacity is as follows:
\begin{IEEEeqnarray}{lr}
C_{\hbox{sum}}=\frac{1}{2}\log\left|I+\sum_{k\in\mathcal{K}}\mathbf{H_{k}Q_{k}H_{k}^{\dag}}\right|-\frac{1}{2}\log\left|I+\sum_{k\in\mathcal{K}}\mathbf{H_{k,e}Q_{k}H_{k,e}^{\dag}}\right|,
\end{IEEEeqnarray}
if the inputs are subject to the following covariance matrices constraints:
\begin{IEEEeqnarray}{lr}
\mathbf{K_{x_{k}}}\leq \mathbf{Q_{k}},~~~~~~~~~~~~~~~~\forall k\in\mathcal{K},
\end{IEEEeqnarray}
where $\mathbf{K_{x_{k}}}$ denotes the covariance matrix of
$\mathbf{x_{k}}$. Theorem \ref{th3} then follows by maximization
over all $\mathbf{Q_{1}}$, $\mathbf{Q_{2}}$ $,...$, and
$\mathbf{Q_{K}}$ that satisfy the power constraint, i.e.,
$Tr(\mathbf{Q_{k}})\leq P$, for all $k\in\mathcal{K}$.

The achievability of this theorem follows from Theorem \ref{cor2} by
choosing $U_{k}=\mathbf{x_{k}}\sim \mathbf{N}(0,\mathbf{Q_{k}})$.
The converse proof is presented in Appendix B.
\end{proof}

According to (\ref{eq13}), it is easy to show that:
\begin{coro}
The total S-DoF for the degraded Gaussian MIMO MAC is $\eta=0$.
\end{coro}
Now consider the general Gaussian MIMO MAC where it's not necessarily degraded. According to Theorem \ref{th1}, by choosing $U_{k}=\mathbf{x_{k}}\sim \mathcal{N}(0,\mathbf{Q_{k}})$, it is easy to see that the following secrecy sum rate is achievable.

\begin{coro}\label{cor1}
For the Gaussian MIMO MAC, an achievable secrecy sum rate is given by
\begin{IEEEeqnarray}{rl}\label{eq14}
R_{\hbox{sum}}=\sum_{k\in\mathcal{K}}R_{k}= \max_{\mathbf{Q}: \mathbf{Q}\succeq 0, Tr(\mathbf{Q})\leq P} \frac{1}{2}\log\left|\mathbf{I}+\sum_{k\in\mathcal{K}}\mathbf{H_{k}Q_{k}H_{k}^{\dag}}\right|-\frac{1}{2}\log\left|\mathbf{I}+\sum_{k\in\mathcal{K}}\mathbf{H_{k,e}Q_{k}H_{k,e}^{\dag}}\right|.
\end{IEEEeqnarray}
\end{coro}
Note that in the above achievable scheme, we chose
$\mathbf{K_{x_{k}}}=\mathbf{Q_{k}}$ which generally is not optimal.
In general, solving the maximization problem of (\ref{eq14}) is
difficult. Reference \cite{20}, however, has solved this problem for
a single antenna case and derived the optimal power control policy.
As shown in \cite{20} even for a single antenna case some users need
to be silent and therefore those users can cooperatively help to jam
the eavesdropper.

We study the S-DoF defined in (\ref{sdf}) to analyze the behavior of
the sum rate in high $SNR$. We design the following strategy scheme
at the transmitters.

To achieve the largest value for S-DoF  we need to separate the
received signals at the legitimate receiver, such that each received
signal has a different dimension in the signal space of the
legitimate receiver. At the same time all the received signals at
the eavesdropper need to be aligned in a minimal subspace of the
signal space of the eavesdropper (see Fig.\ref{f2}).
\begin{figure}
\centerline{\includegraphics[scale=.8]{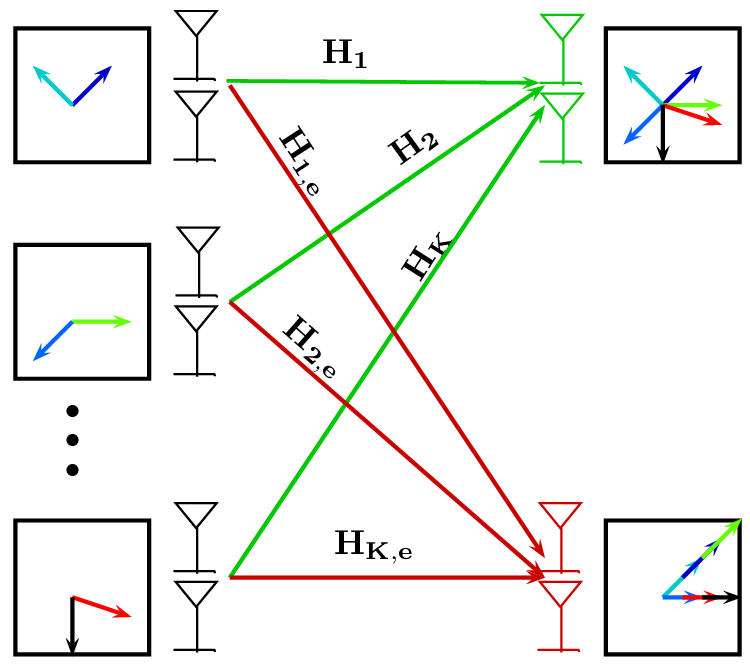}} \caption{Separation/Alignment of Signals at the Intended Receiver/Eavesdropper} \label{f2}
\end{figure}

Let $\mathbf{x_{k}}=\mathbf{F_{k}}\mathbf{v_{k}}$ where, $\mathbf{F_{k}}$ is a pre-coding matrix such that $\mathbf{F_{k}F_{k}^{\dag}}=\mathbf{Q_{k}}$ and $\mathbf{v_{k}}$ is a vector with i.i.d Gaussian components with zero mean and unit variance, i.e, $\mathbf{v_{k}}=[v_{k}^{1},...,v_{k}^{N}]^{\dag}$, such that $v_{k}^{j}\sim \mathcal{N}(0,1)$ for $j=1,2,...,N$. Let $\mbox{\boldmath$\psi$}_{i}=[0,0,...,0,1,0,...,0]^{\dag}$ be a $N\times 1$ vector that all of its elements are zero except the $i$th element which is $1$. Let $\mathbf{F_{k}}=[\mathbf{f^{1}_{k}},\mathbf{f^{2}_{k}},...,\mathbf{f^{N}_{k}}]$, where $\mathbf{f^{i}_{k}}$'s for $i=1,2,...,N$ are  $N\times 1$ vectors that represent the columns of $\mathbf{F_{k}}$. We use the following algorithm to choose $\mathbf{f^{i}_{k}}$'s:
\begin{itemize}
  \item Assume that for users $k=1,2,...,J$, the matrix $\mathbf{H_{k,e}}$ has non-empty null space and the null space of $\mathbf{H_{k,e}}$ for users $k=J+1,J+2,...,K$ is empty. The users $k=1,2,...,\min\{J,N\}$ choose $\mathbf{f^{1}_{k}}$ such that $\mathbf{H_{k,e}}\mathbf{f^{1}_{k}}=0$. Almost surly, we can assume that these vectors occupy separate dimensions at the legitimate receiver.
  \item If $J\geq N$, then all the $N$ dimensions at the legitimate receiver are full and $\eta=N$.
  \item If $J< N$, users $k=J+1,J+2,...,\min\{K,N\}$ then create a vector $\mathbf{f^{1}_{k}}$ such that $\mathbf{H_{k,e}}\mathbf{f^{1}_{k}}=\mbox{\boldmath$\psi$}_{1}$. Theses vectors, therefore, are aligned at the eavesdropper in one dimension and almost surely occupy separate dimensions in the remaining subspace of the legitimate receiver.
  \item If $N\leq K$, at this step all the dimensions at the legitimate receiver are full and one dimension at the eavesdropper has a non-zero signal. Thus, $\eta=N-1$.
  \item If $N>K$, users $k=1,2,...,J$ then also creates a vector $\mathbf{f^{2}_{k}}$ such that $\mathbf{H_{k,e}}\mathbf{f^{2}_{k}}=\mbox{\boldmath$\psi$}_{1}$. These vectors are also aligned at the dimension $\mbox{\boldmath$\psi$}_{1}$ and almost surely $\eta=N-1$
  \item We repeat the above steps such that all the dimensions at the legitimate receiver become full.
\end{itemize}
The above algorithm can be followed  when the users and the
receivers are equipped with different numbers of antennae. The
following theorem characterizes the maximum amount of the total
S-DoF that can be achieved by Gaussian codebooks.
\begin{thm}\label{th3}
For the Gaussian MIMO MAC, the following $\eta$ for S-DoF can almost
surely be achieved for almost all channel gains by using Gaussian
codebooks and under perfect secrecy constraint:
\begin{IEEEeqnarray}{rl}
\eta=\left[\min\left\{\sum_{k\in\mathcal{K}}M_{k},N\right\}-r\right]^{+},
\end{IEEEeqnarray}
where $0\leq r\leq \min\{\sum_{k\in\mathcal{K}}M_{k},N_{e}\}$
depends on the channel gain matrix  $\mathbf{H_{e}}$.
\end{thm}
Note that in Theorem \ref{th3}, it is emphasized that total
$\left[\min\left\{\sum_{k\in\mathcal{K}}M_{k},N\right\}-r\right]^{+}$
S-DoF is achievable for almost all channel gains. It means the set
of all possible gains that the total amount of
$\left[\min\left\{\sum_{k\in\mathcal{K}}M_{k},N\right\}-r\right]^{+}$
S-DoF may not be achieved has the Lebesgue measure zero. In other
words, if all the channel gains are drawn independently from a
random distribution, then almost all of them satisfy properties
required to achieve the total S-DoF, almost surely. The term
``almost surely" means with a probability arbitrary close to 1.
\begin{rem}
In the achievability scheme of Theorem \ref{th3}, all the
transmitted signals are aligned into a $r$ dimensional subspace at
the eavesdropper, and hence, impairs the ability of the eavesdropper
to distinguish any of the secure messages efficiently.
\end{rem}
Now assume that the transmitters can cooperate with each other: we
have a MIMO wiretap channel where the transmitter has
$\sum_{k\in\mathcal{K}}M_{k}$ antennas and the legitimate and
eavesdropper have $N$ and $N_{e}$ antennas, respectively. The
secrecy capacity of this channel is indeed an upper-bound for the
secrecy sum capacity of the Gaussian MIMO MAC. As the capacity of a
non-secret MIMO channel is an upper-bound for the secrecy capacity
of the MIMO wiretap channel, we have the following upper-bound for
S-DoF for the secure Gaussian MIMO MAC.
\begin{lem}
For the Gaussian MIMO MAC, the maximum total achievable S-DoF under
perfect secrecy constraint is given by
\begin{IEEEeqnarray}{rl}
\eta_{\max}=\min\left\{\sum_{k\in\mathcal{K}}M_{k},N\right\}.
\end{IEEEeqnarray}
\end{lem}
The maximum penalty for the achievable S-DoF in Theorem \ref{th3} is
therefore $r$, where, $0\leq r\leq
\min\{\sum_{k\in\mathcal{K}}M_{k},N_{e}\}$. We note that an
achievable S-DoF in the MIMO wiretap channel using zero-forcing
beamforming is given as follows.
\begin{lem}
In the MIMO wiretap channel, the following S-DoF is achievable, almost surely.
\begin{IEEEeqnarray}{rl}
\eta=\min\left\{\left[\sum_{k\in\mathcal{K}}M_{k}-r^{'}\right]^{+},N\right\},
\end{IEEEeqnarray}
\end{lem}
where $1\leq r^{'}\leq N_{e}$ depends on the channel gain matrix  $\mathbf{H_{e}}$.
\begin{rem}
 When the transmitters and the intended receiver are equipped with a sufficiently large number of antenna while the eavesdropper is equipped with a limited number of antenna, then Gaussian codebooks provide a near optimum total of S-DoF for the Gaussian MIMO multiple-access-channel under perfect secrecy constraint.
\end{rem}
Note that when the transmitters and the receivers are equipped with
a single antenna, i.e., $M_{k}=N_{e}=N=1$, then the total achieved
S-DoF is $0$. It should be noted that this result comes from the
lack of enough dimension for signal management at the receivers by
using Gaussian codebooks. In our achievability scheme, nodes
$\mathcal{K}-1$ send sequences from a codebook randomly generated in
an i.i.d. fashion according to a Gaussian distribution. These are
the worst noises from the eavesdropper's perspective if Gaussian
i.i.d. signaling is used in $X_{1}$, see \cite{31}. However, since
the channel is fully connected, $\mathcal{K}-1$ are also the worst
noises for the intended receiver. This effect causes the secrecy
rate to saturate, leading to zero S-DoF. The following Theorem
establishes an upper-bound for the total S-DoF.

\section{Secure DoF of the Single-Antenna Multiple-Access-Channel}
In this section we consider the secure multiple-access-channel of
(\ref{eq0}) when all the transmitters and the receivers have a
single antenna, i.e., $M_{k}=N=N_{e}=1$ for all $k\in\mathcal{K}$.
We showed in the previous section that Gaussian codebooks lead to
zero total S-DoF. Here, we will provide a coding scheme based on
integer codebooks and show that for almost all channel gains a
positive total S-DoF is achievable, almost surely. The following
theorem illustrates our results.
\begin{thm}\label{th4}
For the Gaussian single antenna multiple-access-channel of
(\ref{eq0}) with $M_{k}=N=N_{e}=1$, a total $\frac{K-1}{K}$ secure
degrees-of-freedom can be achieved for almost all channel gains,
almost surely.
\end{thm}
\begin{proof}
When the transmitters and the receivers are equipped with a single
antenna, then the channel model of (\ref{eq0}) is equivalent as
follows:
\begin{IEEEeqnarray}{rl}
Y&=\sum_{k=1}^{K}h_{k}X_{k}+\widetilde{W}_{1}\\ \nonumber
Z&=\sum_{k=1}^{K}h_{k,e}X_{k}+\widetilde{W}_{2},
\end{IEEEeqnarray}
where $\widetilde{W}_{1}\sim \mathcal{N}(0,1)$,
$\widetilde{W}_{2}\sim \mathcal{N}(0,1)$, and $E[X_{k}^{2}]\leq P$
for all $k\in\mathcal{K}$. Let us define
$\widetilde{X}_{k}\stackrel{\triangle}{=}\frac{h_{k,e}}{A}X_{k}$ and
$\widetilde{h}_{k}\stackrel{\triangle}{=}\frac{h_{k}}{h_{k,e}}$ and
without loss of generality assume that $\widetilde{h}_{K}=1$, then
the channel model is equivalent as follows:
\begin{IEEEeqnarray}{rl}
Y&=A\left[\sum_{k=1}^{K-1}\widetilde{h}_{k}\widetilde{X}_{k}+\widetilde{X}_{K}\right]+\widetilde{W}_{1}\\ \nonumber
Z&=A\sum_{k=1}^{K}\widetilde{X}_{k}+\widetilde{W}_{2},
\end{IEEEeqnarray}
where, $A^{2}E[\widetilde{X}_{k}^{2}]\leq\widetilde{P}\stackrel{\triangle}{=} h_{k,e}^{2} P$. In this model we say that the signals are aligned at the eavesdropper according to the following definition:
\begin{defi}
The signals $\widetilde{X}_{1}$, $\widetilde{X}_{2}$,...,$\widetilde{X}_{K}$ are said to be aligned at a receiver
if its received signal is a rational combination of them.
\end{defi}
Note that, in $n$-dimensional Euclidean spaces ($n\geq 2$), two
signals are aligned when they are received in the same direction at
the receiver. In general, $m$ signals are aligned at a receiver if
they span a subspace with dimension less than $m$. The above
definition, however, generalizes the concept of alignment for the
one-dimensional real numbers. Our coding scheme is based on integer
codebooks, which means that $\widetilde{X}_{k}\in \mathbb{Z}$ for
all $k\in\mathcal{K}$. If some integer signals are aligned at a
receiver, then their effect is similar to a single signal at high
SNR regimes. This is due to the fact that rational numbers form a
filed and therefore the sum of constellations from $\mathbb{Q}$ form
a constellation in $\mathbb{Q}$ with an enlarged cardinality.

Before we present our achievability scheme, we need to define the
rational dimension of a set of real numbers.
\begin{defi}(\emph{Rational Dimension})
The rational dimension of a set of real numbers
$\{\widetilde{h}_{1},\widetilde{h}_{2},...,\widetilde{h}_{K-1},\widetilde{h}_{K}=1\}$
is $M$ if there exists a set of real numbers
$\{g_{1},g_{2},...,g_{M}\}$ such that each $\widetilde{h}_{k}$ can
be represented as a rational combination of $g_{i}$'s, i.e.,
$\widetilde{h}_{k}=a_{k,1}g_{1}+a_{k,2}g_{2}+...+a_{k,M}g_{M}$,
where $a_{k,i}\in\mathbb{Q}$ for all $k\in\mathcal{K}$ and
$i\in\mathcal{M}$.
\end{defi}
In fact, the rational dimension of a set of channel gains is the
effective dimension seen at the corresponding receiver. In
particular,
$\{\widetilde{h}_{1},\widetilde{h}_{2},...,\widetilde{h}_{K}\}$ are
\emph{rationally independent} if the rational dimension is $K$,
i.e., none of the $\widetilde{h}_{k}$ can be represented as the
rational combination of other numbers.

Note that all of the channel gains $\widetilde{h}_{k}$ are generated
independently with a distribution. From the number theory, it is
known that the set of all possible channel gains that are rationally
independent have a Lebesgue measure $1$. Therefore, we can assume
that $\{\widetilde{h}_{1},\widetilde{h}_{2},...,\widetilde{h}_{K}\}$
are rationally independent, almost surely. Our achievability coding
scheme is as follows:
\subsubsection{Encoding}
Each transmitter limits its input symbols to a finite set which is
called the transmit constellation. Even though it has access to the
continuum of real numbers, restriction to a finite set has the
benefit of easy and feasible decoding at the intended receiver. The
transmitter $k$ selects a constellation $\mathcal{V}_{k}$ to send
message $W_{k}$. The constellation points are chosen from integer
points, i.e., $\mathcal{V}_{k}\subset\mathbb{Z}$. We assume that
$\mathcal{V}_{k}$ is a bounded set. Hence, there is a constant
$Q_{k}$ such that $\mathcal{V}_{k}\subset [-Q_{k},Q_{k}]$. The
cardinality of $\mathcal{V}_{k}$ which limits the rate of message
$W_{k}$ is denoted by $\|\mathcal{V}_{k}\|$.

Having formed the constellation, the transmitter $k$ constructs a
random codebook for message $W_{k}$  with rate $R_{k}$. This can be
accomplished by choosing a probability distribution on the input
alphabets. The uniform distribution is the first candidate and it is
selected for the sake of simplicity. Therefore, the stochastic
encoder $k$ generates
$2^{n(I(\widetilde{X}_{k};Y|\widetilde{X}_{(\mathcal{K}-k)^{c}})+\epsilon_{k})}$
independent and identically distributed sequences
$\widetilde{x}_{k}^{n}$ according to the distribution
$P(\widetilde{x}_{k}^{n})=\prod_{i=1}^{n}P(\widetilde{x}_{k,i})$,
where $P(\widetilde{x}_{k,i})$ denotes the probability distribution
function of the uniformly distributed random variable
$\widetilde{x}_{k,i}$ over $\mathcal{V}_{k}$. Next, randomly
distribute these sequences into $2^{nR_{k}}$ bins. Index each of the
bins by $w_{k}\in\{1,2,...,2^{nR_{k}}\}$.

For each user $k\in\mathcal{K}$, to send message $w_{k}$, the
transmitter looks for a $\widetilde{x}_{k}^{n}$ in bin $w_{k}$. The
rates are such that there exist more than one
$\widetilde{x}_{k}^{n}$. The transmitter randomly chooses one of
them and sends $x_{k}^{n}=A\frac{\widetilde{x}_{k}^{n}}{h_{k,e}}$.
The parameter $A$ controls the input power.
\subsubsection{Decoding}
At a specific time, the received signal at the legitimate receiver is as follows:
\begin{IEEEeqnarray}{rl}
Y=A\left[\widetilde{h}_{1}\widetilde{X}_{1}+\widetilde{h}_{2}\widetilde{X}_{2}+...+\widetilde{h}_{K-1}\widetilde{X}_{K-1}+\widetilde{X}_{K}\right]+\widetilde{W}_{1}
\end{IEEEeqnarray}
The legitimate receiver passes the received signal $Y$ through a
hard decoder. The hard decoder looks for a point $\widetilde{Y}$ in
the received constellation
$\mathcal{V}_{r}=A\left[\widetilde{h}_{1}\mathcal{V}_{1}+\widetilde{h}_{2}\mathcal{V}_{2}+...+\widetilde{h}_{K-1}\mathcal{V}_{K-1}+\mathcal{V}_{K}\right]$
which is the nearest point to  the received signal $Y$. Therefore,
the continuous channel changes to a discrete one in which the input
symbols are taken from the transmit constellations $\mathcal{V}_{k}$
and the output symbols belongs to the received constellation
$\mathcal{V}_{r}$. $\widetilde{h}_{k}$'s are rationally independent
which means that the equation
$A\left[\widetilde{h}_{1}X_{1}+\widetilde{h}_{2}X_{2}+...+\widetilde{h}_{K-1}X_{K-1}+X_{K}\right]=0$
has no rational solution. This property implies that any real number
$v_{r}$ belonging to the constellation $\mathcal{V}_{r}$ is uniquely
decomposable as $v_{r}=A\sum_{k=1}^{K}
\widetilde{h}_{k}\widehat{\widetilde{X}}_{k}$. Note that if there
exists another possible decomposition
$\widetilde{v}_{r}=A\sum_{k=1}^{K}
\widetilde{h}_{k}\widehat{\widetilde{X}}_{k}^{'}$, then
$\widetilde{h}_{k}$'s have to be rationally-dependent which is a
contradiction. We call this property as property $\Gamma$. This
property in fact implies that if there is no additive noise in the
channel, then the receiver can decode all the transmitted signals
with zero error probability.

\begin{rem}
In a random environment it is easy to show that the set of channels
gains which are rationally-dependent has a measure of zero with
respect to the Lebesgue measure. Therefore, Property $\Gamma$ is
almost surely satisfied.
\end{rem}

\subsubsection{Error Probability Analysis}

Let $d_{\min}$ denote the minimum distance in the received
constellation $\mathcal{V}_{r}$. Having property $\Gamma$, the
receiver can decode the transmitted signals. Let $V_{r}$ and
$\hat{V}_{r}$ be the transmitted and decoded  symbols, respectively.
The probability of error, i.e., $P_{e}=P(\hat{V}_{r}\neq V_{r})$, is
bounded as follows:
\begin{equation}\label{eq10}
P_{e}\leq Q(\frac{d_{\min}}{2})\leq \exp(-\frac{d_{\min}^{2}}{8})
\end{equation}
where
$Q(x)=\frac{1}{\sqrt{2\pi}}\int_{x}^{\infty}\exp(-\frac{t^{2}}{2})dt$.
Note that finding $d_{\min}$ is not easy in general. Using
Khintchine and Groshev theorems, however, it is possible to lower
bound the minimum distance. Here we explain some background
information for using the theorems of Khintchine and Groshev.

The field of Diophantine approximation in number theory deals with
approximation of real numbers with rational numbers. The reader is
referred to \cite{32,33} and the references therein. The Khintchine
theorem is one of the cornerstones in this field. This theorem
provides a criteria for a given function
$\psi:\mathbb{N}\rightarrow\mathbb{R}_{+}$ and real number $h$, such
that $|p + \widetilde{h}q| < \psi(|q|)$ has either infinitely many
solutions or at most, finitely many solutions for
$(p,q)\in\mathbb{Z}^{2}$. Let $\mathcal{A}(\psi)$ denote the set of
real numbers such that $|p +\widetilde{h}q| < \psi(|q|)$ has
infinitely many solutions in integers. The theorem has two parts.
The first part is the convergent part and states that if $\psi(|q|)$
is convergent, i.e.,
\begin{equation}
\sum_{q=1}^{\infty}\psi(q)< \infty
\end{equation}
then $\mathcal{A}(\psi)$ has a measure of zero with respect to the
Lebesque measure. This part can be rephrased in a more convenient
way, as follows. For almost all real numbers, $|p+\widetilde{h}q| >
\psi(|q|)$ holds for all $(p, q) \in\mathbb{Z}^{2}$ except for
finitely many of them. Since the number of integers violating the
inequality is finite, one can find a constant $c$ such that
\begin{equation}
|p + \widetilde{h}q| > c\psi(|q|)
\end{equation}
holds for all integers $p$ and $q$, almost surely. The divergent part of the theorem states that $\mathcal{A}(\psi)$ has the
full measure, i.e. the set $\mathbb{R} - \mathcal{A}(\psi)$ has measure zero, provided that $\psi$ is decreasing and $\psi(|q|)$ is divergent,
i.e.,
\begin{equation}
\sum_{q=1}^{\infty}\psi(q)=\infty.
\end{equation}
There is an extension to Khintchine's theorem which regards the
approximation of linear forms. Let $\mathbf{\widetilde{h}} =
(\widetilde{h}_{1}, \widetilde{h}_{2},...,\widetilde{h}_{K-1})$ and
$\mathbf{q} = (q_{1}, q_{2},... , q_{K-1})$ denote $(K-1)$-tuples in
$\mathbb{R}^{K-1}$ and $\mathbb{Z}^{K-1}$, respectively. Let
$\mathcal{A}_{K-1}(\psi)$ denote the set of $(K-1)$-tuple real
numbers $\mathbf{\widetilde{h}}$ such that
\begin{equation}
|p + q_{1}\widetilde{h}_{1} + q_{2}\widetilde{h}_{2} +... +q_{K-1}\widetilde{h}_{K-1}| <\psi(|\mathbf{q}|_{\infty})
\end{equation}
has infinitely many solutions for $p\in \mathbb{Z}$ and
$\mathbf{q}\in\mathbb{Z}^{K-1}$. Here, $|\mathbf{q}|_{\infty}$ is
the supreme norm of $\mathbf{q}$ which is defined as $\max_{k}
|q_{k}|$. The following theorem illustrates the Lebesque measure of
the set $\mathcal{A}_{K-1}(\psi)$ [].
\begin{thm}(Khintchine-Groshev)
Let $\psi:\mathbb{N}\rightarrow \mathbb{R}_{+}$. Then, the set $\mathcal{A}_{K-1}(\psi)$ has measure zero provided that
\begin{equation}\label{eq9}
\sum_{q=1}^{\infty}q^{K-2}\psi(q)<\infty
\end{equation}
and has the full measure if
\begin{equation}
\sum_{q=1}^{\infty}q^{K-2}\psi(q)=\infty~~~~\hbox{and $\psi$ is monotonic}
\end{equation}
\end{thm}
In this paper, we are interested in the convergent part of the
theorem. Moreover, given an arbitrary $\epsilon > 0$ the function
$\psi(q) =\frac{1}{q^{K-1+\epsilon}}$ satisfies the condition of
(\ref{eq9}). In fact, the convergent part of the above theorem can
be stated as follows. For almost all $K-1$-tuple real numbers
$\mathbf{\widetilde{h}}$ there exists a constant $c$ such that
\begin{equation}
|p + q_{1}\widetilde{h}_{1} + q_{2}\widetilde{h}_{2} + . . . + q_{K-1}\widetilde{h}_{K-1}| >\frac{c}{(\max_{k} |q_{k}|)^{K-1+\epsilon}}
\end{equation}
holds for all $p \in\mathbb{Z}$ and $\mathbf{q}\in\mathbb{Z}^{K-1}$.
The Khintchine-Groshev theorem can be used to bound the minimum
distance of points in the received constellation $\mathcal{V}_{r}$.
In fact, a point in the received constellation has a linear form of
\begin{equation}
v_{r}
=A\left[\widetilde{h}_{1}v_{1}+\widetilde{h}_{2}v_{2}+...+\widetilde{h}_{K-1}v_{K-1}+v_{K}\right],
\end{equation}
Therefore, we can conclude that
\begin{equation}
d_{\min} >\frac{Ac}{(\max_{k\in\{1,2,...,K-1\}}Q_{k})^{K-1+\epsilon}}.
\end{equation}
The probability of error in hard decoding, see (\ref{eq10}), can be bounded as:
\begin{equation}\label{eq11}
P_{e} < \exp\left(-\frac{(Ac)^{2}}{8(\max_{k\in\{1,2,...,K-1\}}Q_{k})^{2K-2+2\epsilon}}\right)
\end{equation}
Let us assume that $Q_{k}$ for all $k\in\{1,2, . . . ,K-1\}$ is
$Q=\lfloor\widetilde{P}^{\frac{1-\epsilon}{2(K+\epsilon)}}\rfloor$.
Moreover, since $E[\widetilde{X}_{k}^{2}]\leq A^{2}Q_{k}^{2}\leq
\widetilde{P}$, we can choose
$A=\widetilde{P}^{\frac{K-1+2\epsilon}{2(K+\epsilon)}}$.
Substituting in (\ref{eq11}) yields
\begin{equation}
P_{e} < \exp (-\frac{c^{2}}{8}\widetilde{P}^{\epsilon}).
\end{equation}
Thus, $P_{e}\rightarrow 0$ when $\widetilde{P}\rightarrow \infty$ or equivalently $P\rightarrow \infty$.
\subsubsection{Equivocation Calculation}
Since the equivocation analysis of Theorem \ref{th1} is valid for
any input distribution, therefore integer inputs satisfy the perfect
secrecy constraint.
\subsubsection{S-DoF Calculation}
The maximum achievable sum rate is as follows:
\begin{IEEEeqnarray}{rl}\label{eq15}
\sum_{k\in\mathcal{K}}R_{k}&= I(\widetilde{X}_{1},\widetilde{X}_{2},...,\widetilde{X}_{K};Y)-I(\widetilde{X}_{1},\widetilde{X}_{2},...,\widetilde{X}_{K};Z)\\ \nonumber
&=H(\widetilde{X}_{1},\widetilde{X}_{2},...,\widetilde{X}_{K}|Z)-H(\widetilde{X}_{1},\widetilde{X}_{2},...,\widetilde{X}_{K}|Y)\\ \nonumber
&\stackrel{(a)}{\geq}H(\widetilde{X}_{1},\widetilde{X}_{2},...,\widetilde{X}_{K}|Z)-1-P_{e}\log\|\widetilde{\mathcal{X}}\|\\ \nonumber
&\stackrel{(b)}{\geq}H(\widetilde{X}_{1},\widetilde{X}_{2},...,\widetilde{X}_{K}|\sum_{k\in\mathcal{K}}\widetilde{X}_{k})-1-P_{e}\log\|\widetilde{\mathcal{X}}\|\\ \nonumber
&\stackrel{(c)}{=}\sum_{k\in\mathcal{K}}H(\widetilde{X}_{k})-H(\sum_{k\in\mathcal{K}}\widetilde{X}_{k})-1-P_{e}\log\|\widetilde{\mathcal{X}}\|\\ \nonumber
&\stackrel{(d)}{=}K\log(2Q+1)-\log(2KQ+1)-1-P_{e}\log\|\widetilde{\mathcal{X}}\|,
\end{IEEEeqnarray}
where $(a)$ follows from Fano's inequality, $(b)$ follows from the
fact that conditioning always decreases entropy, $(c)$ follows from
chain rule, and $(d)$ follows from the fact that $\widetilde{X}_{k}$
has uniform distribution over $\mathcal{V}_{k}=[-Q,Q]$. The S-DoF
therefore can be computed as follows:
\begin{IEEEeqnarray}{rl}
\eta&=\lim_{P\rightarrow \infty}\frac{\sum_{k\in\mathcal{K}}R_{k}}{\frac{1}{2}\log P}\\ \nonumber
&=\frac{(K-1)(1-\epsilon)}{K+\epsilon}
\end{IEEEeqnarray}
Since $\epsilon$ can be arbitrarily small, then $\eta=\frac{K-1}{K}$
is indeed achievable.
\end{proof}
As we saw in the previous section, multiple-antennas (or
equivalently, time-varying and/or frequency-selective channels)
provide enough freedom, which allows us to choose appropriate
signaling directions to separate between the messages at the
intended receiver and at the same time pack the signals into a low
dimensionality subspace at the eavesdropper. In contrary, it was
commonly believed that time-invariant frequency flat single-antenna
channels cannot provide any degrees-of-freedom. In Theorem,
\ref{th4} however, we developed a machinery that transforms the
single-antenna systems to a pseudo multiple-antenna system with some
antennas. The number of available dimensions in the equivalent
pseudo multiple-antenna systems is $K$ when all of $K$ channel gains
between the transmitters and the intended receiver are
rationally-independent (this condition is satisfied almost surely).
The equivalent pseudo multiple-antenna system can simulate the
behavior of a multiple-dimensional system (in time/frequency/space)
and allows us to simultaneously separate the signals at the intended
receiver and align them to the eavesdropper.

Note that in the MISOSE wire-tap channel (Multiple-Input
Single-Output Single-Eavesdropper), when the channel realization is
unique, we can achieve the optimum S-DoF of $1$ through cooperation
among transmitters. This Theorem clarifies the fact that, we lose
the amount of $\frac{1}{K}$ in S-DoF due to the lack of cooperation
between the transmitters. We still gain through the possibility of
signal alignment at the eavesdropper, however.
\subsection{Rationally-Dependent Channel Gains: Multiple-layer coding}
When the channel gains are rationally dependent, then a more
sophisticated multiple-layer constellation design is required to
achieve higher S-DoF. The reason is that some messages share the
same dimension at the intended receiver and as a result, splitting
them requires more structure in constellations. We propose a
multiple-layer constellation that can not only be distinguished at
the intended receiver but are also packed efficiently at the
eavesdropper. This is accomplished by allowing a carry-over from
different levels. In this subsection, for the sake of simplicity, we
consider a two user-secure MAC. This channel is modeled as follows:
\begin{IEEEeqnarray}{rl}
Y&=A\left[\widetilde{h}_{1}\widetilde{X}_{1}+\widetilde{X}_{2}\right]+\widetilde{W}_{1}\\ \nonumber
Z&=A\left[\widetilde{X}_{1}+\widetilde{X}_{2}\right]+\widetilde{W}_{2},
\end{IEEEeqnarray}
where $A$ controls the output power. When the channel gain
$\widetilde{h}_{1}$ is irrational then according to Theorem
\ref{th4} the total S-DoF of $\frac{1}{2}$ is indeed achievable. For
the rational channel gain
$\widetilde{h}_{1}=\frac{n}{m},~n,m\in\mathbb{N}, m\neq 0 $,
however, the coding scheme of Theorem \ref{th4} fails and we need to
use a multiple-layer coding scheme.

In multiple-layer coding scheme, we select the constellation points
in the base $W\in\mathbb{N}$ as follows:
\begin{IEEEeqnarray}{rl}
v(\mathbf{b})=\sum_{l=0}^{L-1}b_{l}W^{l},
\end{IEEEeqnarray}
where, $b_{l}$ for all $l\in\{0,1,...,L-1\}$ are independent random
variables which take value from $\{0,1,2,...,,a-1\}$ with uniform
probability distribution. $\mathbf{b}$ represents the vector
$\mathbf{b}=\{b_{0},b_{1},...,b_{L-1}\}$ and the parameter $a$
controls the number of constellation points. We assume that $a<W$
and therefore, all constellation points are distinct and the size of
the constellations are $|\mathcal{V}_{1}|=|\mathcal{V}_{2}|=a^{L}$.
The maximum possible rate for each user is therefore bounded by
$L\log a$.

At each transmitter a random codebook is generated by randomly
choosing $b_{l}$ according to a uniform distribution. The signal
transmitted by users $1$ and $2$ are
$\widetilde{X}_{1}=v(\mathbf{b})$ and
$\widetilde{X}_{2}=v(\mathbf{b^{'}})$, respectively. Note that the
above multiple-layer constellation had a DC component and this
component needs to be removed at the transmitters. The DC component,
however, duplicated the achievable rates and has no effect on the
S-DoF.

To calculate $A$, since $b_{l}$ and $b_{j}$ are independent for $l \neq j$, we have the following chain of inequalities:
\begin{IEEEeqnarray}{rl}
A^{2}E[\widetilde{X}_{1}^{2}]&=A^{2}W^{2(L-1)}\sum_{l=0}^{L-1}E[b_{l}^{2}]W^{-2l}\\ \nonumber
&\leq A^{2}W^{2(L-1)}\frac{(a-1)(2a-1)}{6}\sum_{l=0}^{\infty}W^{-2l}\\ \nonumber
&\leq A^{2}W^{2(L-1)}\frac{a^{2}}{3}\frac{1}{1-W^{-2}}\\ \nonumber
&\leq \frac{A^{2}a^{2}W^{2L}}{W^{2}-1}.
\end{IEEEeqnarray}
Therefore, by choosing $A=\frac{\sqrt{(W^{2}-1)\widetilde{P}}}{aW^{L}}$ the power constraint $A^{2}E{\widetilde{X_{1}^{2}}}\leq \widetilde{P}$ is satisfied. The received constellation at the intended receiver and the eavesdropper can be written as follows, respectively:
\begin{IEEEeqnarray}{rl}
Y&=\frac{A}{m}\sum_{l=0}^{L-1}\left(nb_{l}+mb^{'}_{l}\right)W^{l}+\widetilde{W}_{1}\\ \nonumber
Z&=A\sum_{l=0}^{L-1}\left(b_{l}+b^{'}_{l}\right)W^{l}+\widetilde{W}_{2}.
\end{IEEEeqnarray}
A point in the received constellation $\mathcal{V}_{r}$ of the intended receiver can be represented as follows:
\begin{IEEEeqnarray}{rl}
v_{r}(\mathbf{b},\mathbf{b^{'}})=\frac{A}{m}\sum_{l=0}^{L-1}\left(nb_{l}+mb^{'}_{l}\right)W^{l}.
\end{IEEEeqnarray}
Note that the received constellation needs to satisfy the property
$\Gamma$, as the intended receiver needs to uniquely decode the
transmitted signals. The following theorem characterizes the total
achievable S-DoF.
\begin{thm}\label{th5}
The following S-DoF is achievable for the two user single antenna MAC with rational channel gain $\widetilde{h}_{1}=\frac{n}{m}$:
\begin{IEEEeqnarray}{rl}
\eta=\left\{
       \begin{array}{ll}
         \frac{\log(n)}{\log(n(2n-1))}, & \hbox{if}~2n\geq m \\
         \frac{\log(s+1)}{\log((s+1)(2s+1))}, & \hbox{if}~2n<m ~\hbox{and}~m=2s+1 \\
         \frac{\log(s)}{\log(2s^{2}-n)}, & \hbox{if}~2n<m ~\hbox{and}~m=2s
       \end{array}
     \right.
\end{IEEEeqnarray}
\end{thm}
\begin{proof}
Let us first assume that the property $\Gamma$ is satisfied for given $W$ and $a$. It is easy to show that the minimum distance in the received constellation $\mathcal{V}_{r}$ is $d_{\min}=\frac{A}{m}$. The probability of error is therefore bounded as follows:
\begin{IEEEeqnarray}{rl}
P_{e}&\leq \exp(-\frac{d^{2}_{\min}}{8})\\ \nonumber
&=\exp\left(-\frac{(W^{2}-1)\widetilde{P}}{8a^{2}m^{2}W^{2L}}\right).
\end{IEEEeqnarray}
Let us choose $L$ as
\begin{IEEEeqnarray}{rl}
L=\left\lfloor\frac{\log(\widetilde{P}^{0.5-\epsilon})}{\log(W)}\right\rfloor,
\end{IEEEeqnarray}
where $\epsilon>0$ is an arbitrary small constant. Clearly, with this choice of $L$, $P_{e}\leq \exp(-\gamma \widetilde{P}^{2\epsilon})$ where $\gamma$ is a constant. Thus, when $SNR\rightarrow\infty$, then $P_{e}\rightarrow 0$. Using (\ref{eq15}), the S-DoF of the system can be derived as follows:
\begin{IEEEeqnarray}{rl}
\eta&=\lim_{\widetilde{P}\rightarrow\infty}\frac{L\log(a)}{\frac{1}{2}\log\widetilde{P}}\\ \nonumber
&=\lim_{\widetilde{P}\rightarrow\infty}\frac{\left\lfloor\frac{\log(\widetilde{P}^{0.5-\epsilon})}{\log(W)}\right\rfloor\log(a)}{\frac{1}{2}\log\widetilde{P}}\\ \nonumber &=(1-2\epsilon)\frac{\log (a)}{\log(W)}.
\end{IEEEeqnarray}
Since $\epsilon$ is an arbitrary small constant, the total S-DoF of the system is
\begin{IEEEeqnarray}{rl}
\eta=\frac{\log(a)}{\log(W)}.
\end{IEEEeqnarray}
This equation implies that to achieve the maximum possible $\eta$, we need to maximize $a$ and minimize $W$ with the constraint that the property $\Gamma$ is satisfied. Table I shows the choices of $a$ and $W$ for Theorem \ref{th5}. To complete the proof we need to show that with the choices of Table I, the property $\Gamma$ is satisfied.
\begin{lem}\label{lm3}
The property $\Gamma$ holds for all the choices of Table I.
\end{lem}
\begin{proof}
Please see Appendix $C$.
\end{proof}
\begin{table}[t]
\caption{ Chosen $a$ and $W$ to satisfy property $\Gamma$}
\centering
\begin{center}
\begin{tabular}{|l|c|c|c|}
  \hline
   & $\widetilde{h}_{1}=\frac{n}{m}$ & $a$ & $W$ \\
\hline
  \hbox{Case 1} & $2n\geq m$ & $n$ & $n(2n-1)$ \\
  \hline
Case 2 & $2n<m~$\hbox{and}$~m=2s+1$ & $s+1$ & $(s+1)(2s+1)$ \\
\hline
  Case 3 & $2n<m~$\hbox{and}$~m=2s$ & $s$ & $2s^{2}-n$ \\
  \hline
\end{tabular}
\end{center}
\end{table}
\end{proof}
Note that this result implies that the total achievable S-DoF by
using integer lattice codes is discontinuous with respect to the
channel coefficients.
\section{Conclusion}
In this paper, we considered a $K$ user secure Gaussian MAC with an
external eavesdropper. We proved an achievable rate region for the
secure discrete memoryless MAC and thereafter we established  the
secrecy sum capacity of the degraded Gaussian MIMO MAC using
Gaussian codebooks. For the non-degraded Gaussian MIMO MAC, we
proposed an algorithm inspired by the interference alignment
technique to achieve the largest possible total S-DoF. When all the
terminals are equipped with single antenna, the Gaussian codebooks
lead to zero S-DoF. Therefore, we proposed a novel secure coding
scheme to achieve positive S-DoF in the single antenna MAC. This
scheme converts the single-antenna system into a multiple-dimension
system with fractional dimensions. The achievability scheme is based
on the alignment of signals into a small sub-space at the
eavesdropper, and the simultaneous separation of the signals at the
intended receiver. We proved that total S-DoF of $\frac{K-1}{K}$ can
be achieved for almost all channel gains which  are rationally
independent. For the rationally dependent channel gains, we
illustrated the power of the multi-layer coding scheme, through an
example channel, to achieve a positive S-DoF. As a function of
channel gains, therefore, we showed that the achievable S-DoF is
discontinues.
\appendix
\subsection{Proof of Theorem \ref{th1}}
1) \textit{Codebook Generation}:
 The structure of the encoder for user $k\in\mathcal{K}$ is
depicted in Fig.
 Fix $P(u_{k})$ and $P(x_{k}|u_{k})$. The stochastic
encoder $k$ generates $2^{n(I(U_{k};Y|U_{(\mathcal{K}-k)^{c}})+\epsilon_{k})}$ independent and
identically distributed sequences $u_{k}^{n}$ according to the
distribution $P(u_{k}^{n})=\prod_{i=1}^{n}P(u_{k,i})$. Next,
randomly distribute these sequences into $2^{nR_{k}}$ bins. Index each of the
bins by $w_{k}\in\{1,2,...,2^{nR_{k}}\}$.

2) \textit{Encoding}: For each user $k\in\mathcal{K}$, to send
message $w_{k}$, the transmitter looks for a $u_{k}^{n}$ in bin
$w_{k}$. The rates are such that there exist more than one
$u_{k}^{n}$. The transmitter randomly chooses one of them and then
generates $x_{k}^{n}$ according to
$P(x_{k}^{n}|u_{k}^{n})=\prod_{i=1}^{n}P(x_{k,i}|u_{k,i})$ and sends
it.

3) \textit{Decoding}: The received signals at the legitimate
receiver, $y^{n}$, is the output of the
channel $P(y^{n}|x_{\mathcal{K}}^{n})=\prod_{i=1}^{n}P(y_{i}|x_{\mathcal{K},i})$.
The legitimate receiver looks for the unique sequence $u_{\mathcal{K}}^{n}$ such
that $(u_{\mathcal{K}}^{n},y^{n})$ is jointly typical and declares the
indices of the bins containing $u_{k}^{n}$ as the messages received.

4) \textit{Error Probability Analysis}: Since the region of
(\ref{eq2}) is a subset of the capacity region of the multiple-access-channel without secrecy constraint, then the error probability analysis is the same as \cite{3} and omitted here.

5) \textit{Equivocation Calculation}: To satisfy the perfect secrecy constraint, we need to prove the requirement of (\ref{eq1}). From $H(W_{\mathcal{K}}|Z^{n})$ we have
\begin{IEEEeqnarray}{rl}
H(W_{\mathcal{K}}|Z^{n}) &= H(W_{\mathcal{K}},Z^{n})-H(Z^{n})\\
\nonumber &= H(W_{\mathcal{K}},U_{\mathcal{K}}^{n},Z^{n})- H(U_{\mathcal{K}}^{n}|W_{\mathcal{K}},Z^{n})-H(Z^{n})\\
\nonumber &=H(W_{\mathcal{K}},U_{\mathcal{K}}^{n})+ H(Z^{n}|W_{\mathcal{K}},U_{\mathcal{K}}^{n})- H(U_{\mathcal{K}}^{n}|W_{\mathcal{K}},Z^{n})-H(Z^{n})\\
\nonumber &\stackrel{(a)}{\geq}
H(W_{\mathcal{K}},U_{\mathcal{K}}^{n})+H(Z^{n}|W_{\mathcal{K}},U_{\mathcal{K}}^{n})-n\epsilon_{n}-H(Z^{n})\\
\nonumber
&\stackrel{(b)}{=}H(W_{\mathcal{K}},U_{\mathcal{K}}^{n})+H(Z^{n}|U_{\mathcal{K}}^{n})-n\epsilon_{n}-H(Z^{n}) \\
\nonumber&\stackrel{(c)}{\geq}
H(U_{\mathcal{K}}^{n})+H(Z^{n}|U_{\mathcal{K}}^{n})- n\epsilon_{n}-H(Z^{n}) \\
\nonumber &= H(U_{\mathcal{K}}^{n})- I(U_{\mathcal{K}}^{n};Z^{n})- n\epsilon_{n}\\
\nonumber &\stackrel{(d)}{\geq} I(U_{\mathcal{K}}^{n};Y^{n})-I(U_{\mathcal{K}}^{n};Z^{n})- n\epsilon_{n}\\
\nonumber &\stackrel{(e)}{\geq} n\sum_{k\in\mathcal{K}}R_{k}-n\epsilon_{n}-n\delta_{1n}-n\delta_{4n}\\ \nonumber &= H(W_{\mathcal{K}})-n\epsilon_{n}-n\delta_{1n}-n\delta_{4n},
\end{IEEEeqnarray}
where $(a)$ follows from Fano's inequality, which states that for
sufficiently large $n$, $H(U_{\mathcal{K}}^{n}|W_{\mathcal{K}},Z^{n})$
$\leq h(P_{we}^{(n)})$ $+nP_{we}^{n}R_{w}\leq n\epsilon_{n}$. Here
$P_{we}^{n}$ denotes the wiretapper's error probability of decoding
$u_{\mathcal{K}}^{n}$ in the case that the bin numbers $w_{\mathcal{K}}$ are known to the eavesdropper and
$R_{w}=I(U_{\mathcal{K}};Z)$.
Since the sum rate is small enough, then $P_{we}^{n}\rightarrow 0$
for sufficiently large $n$. $(b)$ follows from the following Markov
chain: $W_{\mathcal{K}}\rightarrow U_{\mathcal{K}}^{n}\rightarrow$
$Z^{n}$. Hence, we have
$H(Z^{n}|W_{\mathcal{K}},U_{\mathcal{K}}^{n})=H(Z^{n}|U_{\mathcal{K}}^{n})$.
$(c)$ follows from the fact that
$H(W_{\mathcal{K}},U_{\mathcal{K}}^{n})\geq H(U_{\mathcal{K}}^{n})$.
$(d)$ follows from that fact that $H(U_{\mathcal{K}}^{n})\geq
I(U_{\mathcal{K}}^{n};Y^{n})$. $(e)$ follows from the following lemma:

\begin{lem}
Assume $U_{\mathcal{K}}^{n}$, $Y^{n}$ and $Z^{n}$ are generated according to the achievability scheme of Theorem \ref{th1}, then we have,
\begin{IEEEeqnarray}{rl}\label{eq3}
nI(U_{\mathcal{K}};Y)-n\delta_{1n}\leq I(U_{\mathcal{K}}^{n};Y^{n})\leq nI(U_{\mathcal{K}};Y)+n\delta_{2n}\\ \label{eq4}
nI(U_{\mathcal{K}};Z)-n\delta_{3n}\leq I(U_{\mathcal{K}}^{n};Z^{n})\leq nI(U_{\mathcal{K}};Z)+n\delta_{4n},
\end{IEEEeqnarray}
where, $\delta_{1n},\delta_{2n},\delta_{3n},\delta_{4n}\rightarrow 0$ when $n\rightarrow\infty$.
\begin{proof}
We first prove (\ref{eq3}). Let $A^{(n)}_{\epsilon}(P_{U_{\mathcal{K}},Z})$ denote the set of typical sequences $(u_{\mathcal{K}}^{n},z^{n})$ with respect to $P_{U_{\mathcal{K}},Z}$, and
\begin{IEEEeqnarray}{lr}
\zeta=\left\{
        \begin{array}{ll}
          1, & (u_{\mathcal{K}}^{n},z^{n})\in A^{(n)}_{\epsilon} \\
          0, & \hbox{otherwise}
        \end{array}
      \right.
\end{IEEEeqnarray}
be the corresponding indicator function. We expand $I(U_{\mathcal{K}}^{n};Z^{n},\zeta)$ and $I(U_{\mathcal{K}}^{n},\zeta;Z^{n})$ as follow:
\begin{IEEEeqnarray}{rl}
I(U_{\mathcal{K}}^{n};Z^{n},\zeta)&=I(U_{\mathcal{K}}^{n};Z^{n})+I(U_{\mathcal{K}}^{n};\zeta|Z^{n})\\ \nonumber
&=I(U_{\mathcal{K}}^{n};\zeta)+I(U_{\mathcal{K}}^{n};Z^{n}|\zeta),
\end{IEEEeqnarray}
\begin{IEEEeqnarray}{rl}
I(U_{\mathcal{K}}^{n},\zeta;Z^{n})&=I(U_{\mathcal{K}}^{n};Z^{n})+I(\zeta;Z^{n}|U_{\mathcal{K}}^{n})\\ \nonumber
&=I(\zeta;Z^{n})+I(U_{\mathcal{K}}^{n};Z^{n}|\zeta).
\end{IEEEeqnarray}
Therefore, we have
\begin{IEEEeqnarray}{rl}
I(U_{\mathcal{K}}^{n};Zn|\zeta)-I(U_{\mathcal{K}}^{n};\zeta|Z^{n})\leq I(U_{\mathcal{K}}^{n};Z^{n})\leq I(U_{\mathcal{K}}^{n};Z^{n}|\zeta)+I(\zeta;Z^{n}).
\end{IEEEeqnarray}
Note that $I(\zeta;Z^{n})\leq H(\zeta)\leq 1$ and $I(U_{\mathcal{K}}^{n};\zeta|Z^{n})\leq H(\zeta|Z^{n})\leq H(\zeta)\leq 1$. Thus, the above inequality implies that
\begin{IEEEeqnarray}{rl}\label{eq8}
\sum_{j=0}^{1}P(\zeta=j)I(U_{\mathcal{K}}^{n};Zn|\zeta=j)-1\leq
I(U_{\mathcal{K}}^{n};Z^{n})\leq
\sum_{j=0}^{1}P(\zeta=j)I(U_{\mathcal{K}}^{n};Z^{n}|\zeta=j)+1.
\end{IEEEeqnarray}
According to the joint typicality property, we have
\begin{IEEEeqnarray}{rl}\label{eq7}
0\leq P(\zeta=1)I(U_{\mathcal{K}}^{n};Z^{n}|\zeta=1)\leq nP((u_{\mathcal{K}}^{n},z^{n})\in A^{(n)}_{\epsilon}(P_{U_{\mathcal{K}},Z} )) \log\|\mathcal{Z}\|\leq n\epsilon_{n}\log \|\mathcal{Z}\|.
\end{IEEEeqnarray}
Now consider the term $P(\zeta=0)I(U_{\mathcal{K}}^{n};Z^{n}|\zeta=0)$. Following the sequence joint typicality properties, we have
\begin{IEEEeqnarray}{rl}\label{eq6}
(1-\epsilon_{n})I(U_{\mathcal{K}}^{n};Z^{n}|\zeta=0)\leq P(\zeta=0)I(U_{\mathcal{K}}^{n};Z^{n}|\zeta=0)\leq I(U_{\mathcal{K}}^{n};Z^{n}|\zeta=0),
\end{IEEEeqnarray}
where
\begin{IEEEeqnarray}{rl}
I(U_{\mathcal{K}}^{n};Z^{n}|\zeta=0)=\sum_{(u_{\mathcal{K}}^{n},z^{n})\in A^{(n)}_{\epsilon}}P(u_{\mathcal{K}}^{n},z^{n})\log\frac{P(u_{\mathcal{K}}^{n},z^{n})}{P(u_{\mathcal{K}}^{n})P(z^{n})}.
\end{IEEEeqnarray}
Since $H(U_{\mathcal{K}},Z)-\epsilon_{n}\leq -\frac{1}{n}\log P(u_{\mathcal{K}}^{n},z^{n})\leq H(U_{\mathcal{K}},Z)+\epsilon_{n}$, then we have,
\begin{IEEEeqnarray}{rl}
n\left[-H(U_{\mathcal{K}},Z)+H(U_{\mathcal{K}})+H(Z)-3\epsilon_{n}\right]\leq I(U_{\mathcal{K}}^{n};Z^{n}|\zeta=0)\leq n\left[-H(U_{\mathcal{K}},Z)+H(U_{\mathcal{K}})+H(Z)+3\epsilon_{n}\right],
\end{IEEEeqnarray}
or equivalently,
\begin{IEEEeqnarray}{rl}\label{eq5}
n\left[I(U_{\mathcal{K}};Z)-3\epsilon_{n}\right]\leq I(U_{\mathcal{K}}^{n};Z^{n}|\zeta=0)\leq n\left[I(U_{\mathcal{K}};Z)+3\epsilon_{n}\right].
\end{IEEEeqnarray}
By substituting (\ref{eq5}) into (\ref{eq6}) and then substituting (\ref{eq6}) and (\ref{eq7}) into (\ref{eq8}), we get the desired result,
\begin{IEEEeqnarray}{rl}
nI(U_{\mathcal{K}};Z)-n\delta_{1n}\leq I(U_{\mathcal{K}}^{n};Z^{n})\leq nI(U_{\mathcal{K}};Z)+n\delta_{2n},
\end{IEEEeqnarray}
where
\begin{IEEEeqnarray}{lr}
\delta_{1n}=\epsilon_{n}I(U_{\mathcal{K}};Z)+3(1-\epsilon_{n})\epsilon_{n}+\frac{1}{n}\\ \nonumber
\delta_{2n}=3\epsilon_{n}+\epsilon_{n}\log\|\mathcal{Z}\|+\frac{1}{n}.
\end{IEEEeqnarray}
Following the same steps, one can prove (\ref{eq4}).
\end{proof}
\end{lem}
\subsection{Proof of the Converse for Theorem \ref{th2}}
Before starting the proof, we first present some useful lemmas.
\begin{lem}\label{lm1}
The secrecy sum capacity of the Gaussian MIMO MAC is upper-bounded by
\begin{IEEEeqnarray}{rl}
C_{\hbox{sum}}\leq \max_{P(\mathbf{x_{1}})P(\mathbf{x_{2}})...P(\mathbf{\mathbf{x_{K}}})} I(\mathbf{x_{1},x_{2},...,x_{K}};\mathbf{y}|\mathbf{z}),
\end{IEEEeqnarray}
where maximization is over all distributions $P(\mathbf{x_{1}})P(\mathbf{x_{2}})...P(\mathbf{\mathbf{x_{K}}})$ that satisfy the power constraint, i.e., $Tr(\mathbf{x^{\dag}}\mathbf{x})\leq P$.
\end{lem}
\begin{proof}
According to Fano's inequality and the perfect secrecy constraint, we have
\begin{IEEEeqnarray}{rl}
n\sum_{k\in\mathcal{K}}R_{k}&\leq I(W_{\mathcal{K}};\mathbf{y^{n}})-I(W_{\mathcal{K}};\mathbf{z^{n}})\\ \nonumber
&\stackrel{(a)}{\leq} I(W_{\mathcal{K}};\mathbf{y^{n},z^{n}})-I(W_{\mathcal{K}};\mathbf{z^{n}})\\ \nonumber
&\stackrel{(b)}{=}I(W_{\mathcal{K}};\mathbf{y^{n}}|\mathbf{z^{n}})\\ \nonumber
&=h(\mathbf{y^{n}}|\mathbf{z^{n}})-h(\mathbf{y^{n}}|W_{\mathcal{K}},\mathbf{z^{n}})\\ \nonumber
&\stackrel{(c)}{\leq} h(\mathbf{y^{n}}|\mathbf{z^{n}})-h(\mathbf{y^{n}}|W_{\mathcal{K}},\mathbf{x^{n}_{\mathcal{K}}},\mathbf{z^{n}})\\ \nonumber
&\stackrel{(d)}{\leq} h(\mathbf{y^{n}}|\mathbf{z^{n}})-h(\mathbf{y^{n}}|\mathbf{x^{n}_{\mathcal{K}}},\mathbf{z^{n}})\\ \nonumber
&\stackrel{(e)}{\leq} h(\mathbf{y^{n}}|\mathbf{z^{n}})-\sum_{i=1}^{n}h(\mathbf{y(i)}|\mathbf{x_{\mathcal{K}}(i)},\mathbf{z(i)})\\ \nonumber
&\stackrel{(f)}{\leq} \sum_{i=1}^{n}h(\mathbf{y(i)}|\mathbf{z(i)})-h(\mathbf{y(i)}|\mathbf{x_{\mathcal{K}}(i)},\mathbf{z(i)})\\ \nonumber
&\leq nI(\mathbf{x_{\mathcal{K}}};\mathbf{y}|\mathbf{z},q)\\ \nonumber
&\stackrel{(g)}{\leq} nI(\mathbf{x_{\mathcal{K}}};\mathbf{y}|\mathbf{z}),
\end{IEEEeqnarray}
where $(a)$ and $(b)$ follows from chain rule, $(c)$ follows from the fact that conditioning decreases the differential entropy, $(d)$ follows from the Markov chain $W_{\mathcal{K}}\rightarrow (\mathbf{x_{\mathcal{K}}^{n}},\mathbf{z^{n}})\rightarrow \mathbf{y^{n}}$, $(e)$ follows from the fact that the channel is memoryless, $(f)$ is obtained by defining a time-sharing random variable $q$ that
takes values uniformly over the index set $\{1, 2, . . . , n\}$ and
defining $(\mathbf{x_{\mathcal{K}}},\mathbf{y}, \mathbf{z})$ to be the tuple of random variables that
conditioned on $q = i$, have the same joint distribution as
$(\mathbf{x_{\mathcal{K}}(i)}, \mathbf{y(i)},\mathbf{z(i)})$. Finally, $(g)$ follows from the fact
that $I(\mathbf{x_{\mathcal{K}}};\mathbf{y}|\mathbf{z})$ is concave in $P(\mathbf{x_{1}})...P(\mathbf{x_{K}})$ (see, e.g.,[],[9, Appendix I] for a proof), so that Jensen's inequality can be applied.
\end{proof}
\begin{lem}\label{lm2}
If $\mathbf{D}\mathbf{H_{k}}=\mathbf{H_{k,e}}$ for all $k\in\mathcal{K}$ and $\mathbf{DD^{\dag}}\preceq \mathbf{I}$, then the function
\begin{IEEEeqnarray}{rl} f(\mathbf{X_{1}},\mathbf{X_{2}},...,\mathbf{X_{K}})=\frac{1}{2}\log\left|\mathbf{I}+\sum_{k\in\mathcal{K}}\mathbf{H_{k}X_{k}H_{k}^{\dag}}\right|-\frac{1}{2}\log\left|\mathbf{I}+\sum_{k\in\mathcal{K}}\mathbf{H_{k,e}X_{k}H_{k,e}^{\dag}}\right|
\end{IEEEeqnarray}
is a concave function of $(\mathbf{X_{1}},...,\mathbf{X_{K}})$ for $\mathbf{X_{k}}\succeq 0$ for all $k\in\mathcal{K}$. Moreover, for $(\mathbf{X_{1}},...,\mathbf{X_{K}})$ such that $\mathbf{X_{k}}\succeq 0$ and $(\mbox{\boldmath$\Delta$}_{1},...,\mbox{\boldmath$\Delta$}_{K})$ such that $\mbox{\boldmath$\Delta$}_{k}\succeq 0$, we have
\begin{IEEEeqnarray}{rl}\label{eq12}
f(\mathbf{X_{1}},...,\mathbf{X_{K}})\leq
f(\mathbf{X_{1}}+\mbox{\boldmath$\Delta$}_{1},...,\mathbf{X_{K}}+\mbox{\boldmath$\Delta$}_{K}).
\end{IEEEeqnarray}
\end{lem}
\begin{proof}
Using the degradedness property of $\mathbf{D}\mathbf{H_{k}}=\mathbf{H_{k,e}}$, the function $f(.)$ can be re-written as follows:
\begin{IEEEeqnarray}{rl} f(\mathbf{X_{1}},\mathbf{X_{2}},...,\mathbf{X_{K}})&=\frac{1}{2}\log\left|\mathbf{I}+\sum_{k\in\mathcal{K}}\mathbf{H_{k}X_{k}H_{k}^{\dag}}\right|-\frac{1}{2}\log\left|\mathbf{I}+\sum_{k\in\mathcal{K}}\mathbf{H_{k,e}X_{k}H_{k,e}^{\dag}}\right|\\ \nonumber &= \frac{1}{2}\log\left|\mathbf{I}+\sum_{k\in\mathcal{K}}\mathbf{H_{k}X_{k}H_{k}^{\dag}}\right|-\frac{1}{2}\log\left|\mathbf{I}+\sum_{k\in\mathcal{K}}\mathbf{DH_{k}X_{k}H_{k}^{\dag}D^{\dag}}\right|\\ \label{eq13} &\stackrel{(a)}{=}\frac{1}{2}\log\frac{\left|\mathbf{I}+\sum_{k\in\mathcal{K}}\mathbf{H_{k}X_{k}H_{k}^{\dag}}\right|}{\left|\left[\left(\mathbf{D^{\dag}D}\right)^{-1}-\mathbf{I}\right]+\left[\mathbf{I}+\sum_{k\in\mathcal{K}}\mathbf{H_{k}X_{k}H_{k}^{\dag}}\right]\right|\left|\mathbf{D^{\dag}D}\right|},
\end{IEEEeqnarray}
where $(a)$ follows from the fact that $\left|\mathbf{I}+\mathbf{AB}\right|=\left|\mathbf{I}+\mathbf{BA}\right|$. According to [23, Lemma II.3], this function is concave with regard to $\mathbf{I}+\sum_{k\in\mathcal{K}}\mathbf{H_{k}X_{k}H_{k}^{\dag}}$, and is therefore concave with regard to $(\mathbf{X_{1}},\mathbf{X_{2}},...,\mathbf{X_{K}})$.

To prove the property of (\ref{eq12}), note that for any
$\mathbf{A}\succeq 0$, $\mbox{\boldmath$\Delta$}\succeq 0$ and
$\mathbf{B}\succ 0$, we have the following property [24,p.3942]:
\begin{IEEEeqnarray}{rl}
\frac{|\mathbf{B}|}{|\mathbf{A}+\mathbf{B}|}\leq \frac{|\mathbf{B}+\mbox{\boldmath$\Delta$}|}{|\mathbf{A}+\mathbf{B}+\mbox{\boldmath$\Delta$}|}.
\end{IEEEeqnarray}
We choose
$\mbox{\boldmath$\Delta$}=\sum_{k\in\mathcal{K}}\mathbf{H_{k}}\mbox{\boldmath$\Delta$}_{k}\mathbf{H_{k}^{\dag}}$
and apply the above property to $(\ref{eq13})$. We thus obtain,
\begin{IEEEeqnarray}{rl}
f(\mathbf{X_{1}},\mathbf{X_{2}},...,\mathbf{X_{K}})&=\frac{1}{2}\log\frac{\left|\mathbf{I}+\sum_{k\in\mathcal{K}}\mathbf{H_{k}X_{k}H_{k}^{\dag}}\right|}{\left|\left[\left(\mathbf{D^{\dag}D}\right)^{-1}-\mathbf{I}\right]+\left[\mathbf{I}+\sum_{k\in\mathcal{K}}\mathbf{H_{k}X_{k}H_{k}^{\dag}}\right]\right|\left|\mathbf{D^{\dag}D}\right|}\\ \nonumber &\leq \frac{1}{2}\log\frac{\left|\mathbf{I}+\sum_{k\in\mathcal{K}}\mathbf{H_{k}\left(X_{k}+\mbox{\boldmath$\Delta$}_{k}\right)H_{k}^{\dag}}\right|}{\left|\left[\left(\mathbf{D^{\dag}D}\right)^{-1}-\mathbf{I}\right]+\left[\mathbf{I}+\sum_{k\in\mathcal{K}}\mathbf{H_{k}\left(X_{k}+\mbox{\boldmath$\Delta$}_{k}\right)H_{k}^{\dag}}\right]\right|\left|\mathbf{D^{\dag}D}\right|}\\ \nonumber &=f(\mathbf{X_{1}}+\mbox{\boldmath$\Delta$}_{1},...,\mathbf{X_{K}}+\mbox{\boldmath$\Delta$}_{K}).
\end{IEEEeqnarray}
\end{proof}
To prove the converse part, we first start with Lemma \ref{lm1} to
bound the sum rate as follows:
\begin{IEEEeqnarray}{rl}
\sum_{k\in\mathcal{K}}R_{k}&\leq I(\mathbf{x_{\mathcal{K}}};\mathbf{y}|\mathbf{z})\\ \nonumber
&=h(\mathbf{y}|\mathbf{z})-h(\mathbf{y}|\mathbf{x_{\mathcal{K}},z})\\ \nonumber
&=h(\mathbf{y}|\mathbf{z})-h(\mathbf{n_{1}})\\ \nonumber
&\stackrel{(a)}{=}\frac{1}{2}\log\left|\mathbf{I}+\sum_{k\in\mathcal{K}}\mathbf{H_{k}K_{x_{k}}H_{k}^{\dag}}\right|-\frac{1}{2}\log\left|\mathbf{I}+\sum_{k\in\mathcal{K}}\mathbf{H_{k,e}K_{x_{k}}H_{k,e}^{\dag}}\right|\\ \nonumber &\stackrel{(b)}{\leq}\frac{1}{2}\log\left|\mathbf{I}+\sum_{k\in\mathcal{K}}\mathbf{H_{k}Q_{k}H_{k}^{\dag}}\right|-\frac{1}{2}\log\left|\mathbf{I}+\sum_{k\in\mathcal{K}}\mathbf{H_{k,e}Q_{k}H_{k,e}^{\dag}}\right|,
\end{IEEEeqnarray}
where $(a)$ follows from the fact that $h(\mathbf{y}|\mathbf{z})$ is maximized by jointly Gaussian $\mathbf{y}$ and $\mathbf{z}$ for fixed covariance matrix $\mathbf{Q_{y,z}}$ and $(b)$ follows from the degradedness assumption and therefore concavity and monotonicity properties given in Lemma \ref{lm2} and the fact that $\mathbf{K_{x_{k}}}\preceq \mathbf{Q_{k}}$.
\subsection{Proof of Lemma \ref{lm3}}
We prove this lemma by induction on $L$. To show the lemma for $L=0$, it is sufficient to prove that the equation
\begin{IEEEeqnarray}{rl}\label{eq16}
n(b_{0}-\widetilde{b}_{0})+m(b^{'}_{0}-\widetilde{b}^{'}_{0})=0,
\end{IEEEeqnarray}
has no nontrivial solution when
$b_{0},b^{'}_{0},\widetilde{b}_{0},\widetilde{b}^{'}_{0}\in\{0,1,2,...,a-1\}$.
The necessary and sufficient conditions for the equation
(\ref{eq16}) are that $b^{'}_{0}-\widetilde{b}^{'}_{0}$ is dividable
by $n$ \textbf{and} $b_{0}-\widetilde{b}_{0}$ is dividable by $n$.
We show that one of these conditions does not hold for all choices
of Table I.

\emph{Case $1$}: In this case $a=n$. Following the fact that
$-(n-1)\leq b^{'}_{0}-\widetilde{b}^{'}_{0}\leq n-1$, it is easy to
deduce that $n \nmid (b^{'}_{0}-\widetilde{b}^{'}_{0})$.,

\emph{Case $2$}: In this case $a=s+1$ where $m=2s+1$. Following the
fact that $-(2s+1)\leq b_{0}-\widetilde{b}_{0}\leq 2s+1$, it is easy
to deduce that $m \nmid b_{0}-\widetilde{b}_{0}$.

\emph{Case $3$}: In this case $a=s$ where $m=2s$. Following the fact
that $-2s\leq b_{0}-\widetilde{b}_{0}\leq 2s$, it is easy to show
that $m \nmid b_{0}-\widetilde{b}_{0}$.

Now assume that property $\Gamma$ holds for $L-1$. We need to show that the equation
\begin{IEEEeqnarray}{rl}
\frac{A}{m}\sum_{l=0}^{L-1}\left(n(b_{l}-\widetilde{b}_{l})+m(b^{'}_{l}-\widetilde{b}^{'}_{l})\right)W^{l}=0,
\end{IEEEeqnarray}
has no nontrivial solution. Equivalently, this equation can be written as follows:
\begin{IEEEeqnarray}{rl}\label{eq17}
n(b_{0}-\widetilde{b}_{0})+m(b^{'}_{0}-\widetilde{b}^{'}_{0})=W\sum_{l=0}^{L-2}\left(n(b_{l+1}-\widetilde{b}_{l+1})+m(b^{'}_{l+1}-\widetilde{b}^{'}_{l+1})\right)W^{l}.
\end{IEEEeqnarray}
We prove that the above equation has no nontrivial solution in two
steps. First, assume that the right hand side of (\ref{eq17}) is
zero. This equation therefore reduces to
\begin{IEEEeqnarray}{rl}
n(b_{0}-\widetilde{b}_{0})+m(b^{'}_{0}-\widetilde{b}^{'}_{0})=0,
\end{IEEEeqnarray}
which we have already shown that has no nontrivial solution for all
three cases.

Secondly, assume that the right side of (\ref{eq17}) is non-zero.
The equation (\ref{eq17}) can therefore be written as follows:
\begin{IEEEeqnarray}{rl}\label{eq18}
n(b_{0}-\widetilde{b}_{0})+m(b^{'}_{0}-\widetilde{b}^{'}_{0})=cW,
\end{IEEEeqnarray}
where $c\in\mathbb{Z}$ and $c\neq 0$. We need to prove that equation (\ref{eq18}) has no nontrivial solution for all three cases.

\emph{Case $1$}: In this case $W=n(2n-1)$ and $n$ divides
$n(b_{0}-\widetilde{b}_{0})$ and $cW$ as well. However, as $(m,n)=1$
and $-(n-1)\leq b_{0}-\widetilde{b}_{0}\leq n-1$, the equation
(\ref{eq18}) has a solution iff $b^{'}_{0}=\widetilde{b}^{'}_{0}$
which is a contradiction with the fact that
$n|b_{0}-\widetilde{b}_{0}|<n(n-1)<|c|n(2n-1)=|c|W$.

\emph{Case $2$}: In this case $W=(s+1)(2s+1)$, $n=s+1$ and $m=2s+1$.
Thus, $2s+1$ divides both $m(b^{'}_{0}-\widetilde{b}^{'}_{0})$ and
$cW$. Since $(2n,m=2s+1)=1$ and $-2s\leq b_{0}-\widetilde{b}_{0}\leq
2s$, therefore, $2s+1$ cannot divide $n(b_{0}-\widetilde{b}_{0})$.
Hence, equation (\ref{eq18}) has a solution iff
$b_{0}=\widetilde{b}_{0}$ which contradicts the fact that
$m|b^{'}_{0}-\widetilde{b}^{'}_{0}|<|c|W$.

\emph{Case $3$}: In this case $W=2s^{2}-n$, $a=s$, $2n<m$ and
$m=2s$. We have
\begin{IEEEeqnarray}{rl}
|n(b_{0}-\widetilde{b}_{0})+m(b^{'}_{0}-\widetilde{b}^{'}_{0})|&<m|b_{0}-\widetilde{b}_{0}+b^{'}_{0}-\widetilde{b}^{'}_{0}|\\
\nonumber &\leq2m(a-1)\\ \nonumber &=4s(s-1)\\ \nonumber &<2W
\end{IEEEeqnarray}
 and therefore, it suffices to assume $c=1$. Substituting
 $W=2s^{2}-n$ in (\ref{eq18}), we have the following equation:
 \begin{IEEEeqnarray}{rl}\label{eq19}
n(b_{0}-\widetilde{b}_{0}+1)+2s(b^{'}_{0}-\widetilde{b}^{'}_{0})=2s^{2}.
 \end{IEEEeqnarray}
Obviously, $2s$ divides $2s(b^{'}_{0}-\widetilde{b}^{'}_{0})$ and
$=2s^{2}$. However, since $(2s,n)=1$ and $-(2s-1)\leq
b_{0}-\widetilde{b}_{0}\leq 2s-1$, equation (\ref{eq19}) has a
solution iff $b_{0}=\widetilde{b}_{0}-1$ which is impossible due to
the fact $2s|b^{'}_{0}-\widetilde{b}^{'}_{0}|<2s^{2}$. This
completes the proof.

\end{document}